\documentclass[12pt]{article}
\usepackage{amsmath,amssymb}
\usepackage{graphicx}
\usepackage{amsthm}
\usepackage{dsfont}
\usepackage{caption}

\def\E{\mathbf{E}}

\def\1{\mathbf{1}}

\def\al{\alpha}
\def\be{\beta}
\def\pa{\partial}

\def\I{\mathds{1}}
\def\II{\I_{[\frac{\alpha }{A \beta },+\infty)}(n)}
\def\III{\I_{[\frac{\alpha }{A \beta },+\infty)}(N)}
\def\del{\al/\be}
\def\sig{\frac{\al}{\be s}}

\newtheorem{theorem}{Theorem}[section]
\newtheorem{lemma}{Lemma}[section]

\newtheorem{remark}{Remark}

\newcommand{\si}{\sigma}

\def\t{\thanks{Lomonosov Moscow State University, Higher School of Economics, Professor Emeritus from
the University of Warwick, Email: v.n.kolokoltsov@gmail.com}}
\def\tt{\thanks{Moscow Institute of Physics and Technology, Email: dranov.em@phystech.edu}}
\def\ttt{\thanks{Lomonosov Moscow State University, Email: denis.piskun.01@mail.ru}}

\begin{document}
\title{A new approach to the theory of optimal income tax}

\author{
Vassili N. Kolokoltsov\t\,,
Egor M. Dranov\tt\,,
Denis E. Piskun\ttt\\
The project is supported by the Vega Institute Foundation.
}

\maketitle

\begin{abstract}
The Nobel-price winning Mirrlees' theory of optimal taxation inspired a long sequence of research
on its refinement and enhancement. However, an issue of concern has been always the fact that, as was shown in many publications,
the optimal schedule in Mirrlees' paradigm of maximising the total utility (constructed from individually optimised individual ones)
usually did not lead to progressive taxation (contradicting the ethically supported practice in developed economies),
and often even assigned minimal tax rates to the higher paid strata of society.
The first objective of this paper is to support this conclusion by proving a theorem on optimal tax schedule
 in (practically most exploited) piecewise-linear environment under a simplest natural utility function.
 The second objective is to suggest a new paradigm for optimal taxation, where instead of just total average
 utility maximization one introduces a standard deviation of utility as a second parameter
 (in some analogy with Marcowitz portfolio optimization). We show that this approach leads to transparent and easy
  interpreted optimality criteria for income tax.
\end{abstract}


\section{Introduction}

The Nobel-price winning Mirrlees' theory of optimal taxation \cite{Mirr71} inspired a long sequence of research
on its refinement and enhancement. In particular, starting from \cite{Sheshinski}, lots of efforts were devoted
to a simpler model with linear taxation. More advanced and practically possibly most interesting case deals with
piecewise linear tax function, see \cite{Sheshinski89}, \cite{Apps} and \cite{Slem}.
However, an issue of concern (see e.g. \cite{WW}) has been always the fact that, as was shown in many publications,
the optimal schedule in Mirrlees' paradigm of maximising the total utility (constructed from individually optimised individual ones)
usually did not lead to progressive taxation (contradicting the ethically supported practice in developed economies).
Often it even assigned minimal tax rates to the higher paid strata of society, see e.g. \cite{Sadka}, \cite{Seade}.
Several suggestions  were made to circumvent these conclusions, see e.g. \cite{Stiglitz}, \cite{WW}, \cite{Saez}.
A different line of research is connected with optimal taxation in connection with possible ta evasion and corruption,
for which we refer to \cite{KolMal} and \cite{Vasin} and numerous references therein.

The first objective of the present paper is to support the conclusion of not progressive optimal
taxation in the standard model  by proving a simple result on optimal tax schedule
 in piecewise-linear environment under the simplest natural utility function.
 The second and the main objective is to suggest a new paradigm for optimal taxation, where instead of just total average
 utility maximization one introduces the standard deviation of utility as a second parameter
 (in some analogy with Marcowitz portfolio optimization). We show that this approach leads to the transparent and easy
  interpreted optimality criteria for income tax.

The paper is organised as follows. In Section \ref{secstand} we prove the result that, under the simplest
personal utility function $u(c,l)=c-l^2/2$ and linear or piecewise linear (with one kink) tax function, the optimal taxation
in the standard setting of maximal total utility is simply no taxation at all. In Section \ref{newparlin}
we introduce our new paradigm of optimal taxation that takes into account not only the total utility but also
its variance. We perform complete analysis of the problem in the simplest case of linear tax. The result is a very natural
(and easy to interpret) optimal linear taxes depending on the choice of a parameter that quantitatively assesses
the acceptable ratio between the growth of total utility and the growth of its standard deviation, the latter considered
as a measure of social tension in a society. In Section
\ref{newpartwobrac} the analysis is extended to the case of piecewise-linear tax with two brackets. Here analytical
solutions are much more difficult to obtain, but numeric analysis is not too complicated to perform, leading essentially
to analogous results, as in the linear case.  In Section \ref{seclogut} we show how our program of two-criteria
optimal taxation works with another example of personal utility, the logarithmic one.

\section{Piecewise-linear tax in the standard setting}
\label{secstand}

Recall that in the standard approach to optimal income tax (see \cite{Mirr71} and \cite{Sheshinski}),
 it is assumed that individuals are characterised by a parameter $n\ge 0$, yielding their production
 per unit of efforts (e.g. time), with a (exogenously given) distribution $f(n)$. Therefore, if an
 individual chooses a level of efforts $l$, then his income before tax is $y=nl$. Mechanism design
 of a government is specified by a tax function $t(y)$, which is assumed to be continuous and nondecreasing,
so that the income of an individual after tax is $c(nl)=y-t(y)=nl-t(nl)$. Individuals are supposed to choose $l$
that maximises certain (exogenously given) utility function $u(c,l)$ that is assumed to be increasing in $c$ and
 decreasing in $l$. We denote by $l_{max}(n)$ a point of maximum of $u(c(nl),l)$ for an individual parametrized
 by $n$. For a utility discussed below, a finite point $l_{max}(n)$ always exists, though may be not unique for
some discrete set of $n$.

Of course, the final quantitative result depends essentially on $u$, even if one specifies its additional
natural properties, as proposed e.g. in \cite{Mirr71} and \cite{Sheshinski}. The simplest reasonable $u$ can be
chosen as $u(c,l)=c-D(l)$ with an increasing convex function $D(l)$, which we adopt here. Moreover, as can be seen
from our exposition below, the choice of the function $D$ is not very essential, and thus we shall use the utility function
\[
u(c,l)=c-l^2/2.
\]
When $l_{\max}(n)$ are chosen the total expected utility becomes
\begin{equation}
\label{eqtotu}
U=\int_0^{\infty} u(c(nl_{\max}(n)), l_{\max}(n)) f(n) dn
=\int_0^{\infty} [nl_{\max}(n) -t(nl_{\max}(n))- \frac12 l_{\max}^2(n)] f(n) dn.
\end{equation}
Since $l_{\max}$ are not unique only on a discrete set of $n$,
this non-uniqueness does not affect the value of $U$.

 We also assume, as in \cite{Sheshinski}, that the tax is fully distributed, so that
\begin{equation}
\label{eqtaxdist}
\int_0^{\infty} t(y) f(n) dn =\int_0^{\infty} t(nl_{\max}(n)) f(n) dn=0.
\end{equation}

In the standard approach, the objective of the government is to find a function $t(y)$ that maximises \eqref{eqtotu}
under constraint \eqref{eqtaxdist}.

We shall work with the piecewise-linear tax environment suggested in \cite{Sheshinski89}
and further analysed e.g. in \cite{Apps}, \cite{Slem}, \cite{Bogachev}.

In this paper we shall reduce our attention to the case of one bracket (linear case)
and two brackets (one kink). An extension to arbitrary number of brackets will be considered elsewhere.

In case of two brackets the tax on income $y$ is given by the piecewise linear function
\begin{equation}
\label{eqtaxonekink}
t(y)=
\left\{
\begin{aligned}
&-\al +(1-\be_1)y, \quad y<y_1 \\
& -\al +(1-\be_1)y_1+(1-\be_2)(y-y_1), \quad y>y_1,
\end{aligned}
\right.
\end{equation}
with one kink at $y_1$. Here $\be_1, \be_2 \in [0,1]$ specifies the proportion of income
received after tax for incomes $y<y_1$ and $y>y_1$ respectively, and $\al>0$ is a non
income-dependent subsidy. In this setting the government optimization goes not over an
infinite-dimensional set of functions $t(y)$, but over a finite set of $\be_1, \be_2, y_1, \al$.

When $\be_1=\be_2=\be$, the tax function \eqref{eqtaxonekink} becomes linear:
\begin{equation}
\label{eqtaxlin}
t(y)=-\al +(1-\be)y
\end{equation}

\begin{theorem}
\label{thonekink}
The maximum of (\ref{eqtotu}) over the parameter set  $\al\ge 0, \be_1\in [0,1], \be_2\in [0,1], y_1\ge 0$
of under constraint \eqref{eqtaxdist} is realised on the linear case with $\al=0, \be_1=\be_2=1$, so with no tax at all.
 \end{theorem}

 \begin{proof}
For $t(y)$ of form \eqref{eqtaxonekink} we have that
\[
u(c(ln),l)=\left\{
\begin{aligned}
& u_1(l,n)=\al +\be_1 nl-\frac12 l^2, \quad l<y_1/n \\
& u_2(l,n)) =\al +(\be_1-\be_2)y_1 +\be_2 nl -\frac12 l^2, \quad l>y_1/n
\end{aligned}
\right.
\]
Hence the points of maximum for $u_1(l,n)$ and $u_2(l,n)$ are
\[
l_{1m}=\min(\be_1n, y_1/n), \quad l_{2m}=\max(\be_2n, y_1/n).
\]

The key points that distinguish different maxima are
 \[
 n_1=\sqrt{y_1/\be_1}, \quad n_2=\sqrt{y_1/\be_2},
 \quad n_3=\sqrt{  \frac{2y_1}{\be_1+\be_2}}.
 \]

The following result specifies the optimal choice of individuals
under $t(y)$ of type \eqref{eqtaxonekink}. We omit a proof,
as it is obtained by a straightforward analysis
of two quadratic function. It can be seen also as a
particular performance of a more general result from \cite{Bogachev}.

\begin{lemma}
\label{lemmaonekink}
(i) Let $\be_1\le \be_2$ (or $n_1 \ge n_2$:  concave tax). Then
the maximum points $l_{\max}(n)$ and the corresponding $u_{\max}(n)$ and $t_{\max}(n)$ are
\[
\left\{
\begin{aligned}
& l_{\max}(n)=\be_1 n \\
& u_{\max}(n)=\al+\frac12 \be_1^2n^2 \\
& t_{\max}(n)=-\al+(1-\be_1)\be_1 n^2
\end{aligned}
\right.
\]
when $n\le n_3$, or
 \[
\left\{
\begin{aligned}
& l_{\max}(n)=\be_2 n \\
& u_{\max}(n)=\al+\frac12 \be_2^2n^2 -y_1(\be_2-\be_1) \\
& t_{\max}(n)=-\al+(1-\be_2)\be_2 n^2+(\be_2-\be_1)y_1
\end{aligned}
\right.
\]
when $n\ge n_3$.

(ii) Let $\be_1 \ge \be_2$ (or $n_1 \le n_2$: convex tax). Then
\[
l_{\max}(n)=
\left\{
\begin{aligned}
& l_{\max}(n)=\be_1 n \\
& u_{\max}(n)=\al+\frac12 \be_1^2n^2 \\
& t_{\max}(n)=-\al+(1-\be_1)\be_1 n^2
\end{aligned}
\right.
\]
when $n<n_1$,

 \[
\left\{
\begin{aligned}
& l_{\max}(n)=\be_2 n \\
& u_{\max}(n)=\al+\frac12 \be_2^2n^2 -y_1(\be_2-\be_1) \\
& t_{\max}(n)=-\al+(1-\be_2)\be_2 n^2+(\be_2-\be_1)y_1
\end{aligned}
\right.
\]
when $n \ge n_2$, and
 \[
\left\{
\begin{aligned}
& l_{\max}(n)=y_1/ n \\
& u_{\max}(n)=\al+\be_1 y_1 -\frac{y_1^2}{2n^2} \\
& t_{\max}(n)=-\al+(1-\be_1)y_1
\end{aligned}
\right.
\]
when $n\in [n_1, n_2]$.
\end{lemma}

Recall that we are analysing the optimization problem of finding $\al\ge 0, \be_1\in [0,1]$,
$\be_2\in [0,2], y_1\ge 0$ that maximise the total expected utility
\[
U=\int_0^{\infty} u_{\max}(n) f(n) dn
\]
under the constraint
\[
\int_0^{\infty} t(nl_{\max}(n)) f(n) dn =0.
\]

If $\be_1\le \be_2$, the constraint takes the form
\[
\al=\int_0^{n_3} (1-\be_1)\be_1 n^2 f(n) dn
\]
\begin{equation}
\label{eqaltwo1}
+\int_{n_3}^{\infty} [(1-\be_2)\be_2 n^2+(\be_2-\be_1)y_1] f(n) dn
\end{equation}
If
 $\be_1\ge \be_2$, the constraint takes the form
\[
\al=\int_0^{n_1} (1-\be_1)\be_1 n^2 f(n) dn
+(1-\be_1)y_1 \int_{\be_1}^{\be_2} f(n) dn
\]
\begin{equation}
\label{eqaltwo2}
+\int_{n_2}^{\infty} [(1-\be_2)\be_2 n^2+(\be_2-\be_1)y_1] f(n) dn.
\end{equation}

Substituting we get, in the first case $\be_1\le \be_2$, that
\[
U=\int u_{\max}(n) f(n) dn =\int_0^{n_3} (\be_1-\frac12 \be_1^2)n^2 f(n) dn
+\int_{n_3}^{\infty}(\be_2-\frac12 \be_2^2) n^2 f(n) dn.
\]

And then
\[
\frac{\pa U}{\pa \be_2}
=\frac{\pa n_3}{\pa \be_2}(\be_1-\be_2)(1-\frac12 (\be_1+\be_2)) n_3^2 f(n_3)
\]
\[
 +\int_{n_3}^{\infty} (1-\be_2) n^2 f(n) dn,
\]
\[
\frac{\pa U}{\pa \be_1}
=\frac{\pa n_3}{\pa \be_1}(\be_1-\be_2)(1-\frac12 (\be_1+\be_2)) n_3^2 f(n_3)
\]
\[
 +\int_0^{n_3} (1-\be_1) n^2 f(n) dn,
\]
which both are positive.

In case $\be_1\ge \be_2$,
\[
U =\int_0^{n_1} (\be_1-\frac12 \be_1^2)n^2 f(n) dn
\]
\[
+\int_{n_2}^{\infty}(\be_2-\frac12 \be_2^2) n^2 f(n) dn
+\int_{n_1}^{n_2} (y_1-\frac{y_1^2}{2n^2}) f(n) dn.
\]

And then
\[
\frac{\pa U}{\pa \be_2}
=\frac{\pa n_2}{\pa \be_2}[(y_1-\frac{y_1^2}{2n_2^2})-(\be_2-\frac12 \be_2^2)n_2^2] f(n_2)
\]
\[
 +\int_{n_2}^{\infty} (1-\be_2) n^2 f(n) dn,
\]
\[
\frac{\pa U}{\pa \be_1}
=-\frac{\pa n_1}{\pa \be_1}[(y_1-\frac{y_1^2}{2n_1^2})-(\be_1-\frac12 \be_1^2)n_1^2] f(n_1)
\]
\[
 +\int_0^{n_1} (1-\be_1) n^2 f(n) dn,
\]
which again both are positive, because the first terms vanish.
Therefore, the maximum is realised on $\be_2=\be_1=1$.

\end{proof}

\section{New paradigm: linear tax}
\label{newparlin}

Taking partial inspiration from the Marcowitz portfolio theory, we suggest here a new paradigm
for the optimal taxation, where instead of the usual total utility optimization, we suggest a
two-criteria optimization problem, choosing as a second criteria the variance or the standard deviation
of the utility function. While in Marcowitz theory the variance is taken as a measure of risk
(with the usual criticism that only deviations in the undesired direction are related to risk),
in our setting the variance expresses the level of inequality in a society that can be also interpreted as
the level of social tension. In the same way as the optimal portfolio theory can be (and was) recast
alternatively using different measures of risk, our theory can be also modified by using other measures
of social inequality, for a modern review of such measures we refer to \cite{Haye}.

In order to see how it works, let us start with the simple case of linear tax \eqref{eqtaxlin}
and the simplest utility function $u=c-l^2/2$, where all calculations can be performed explicitly.

Then $c=\al+\be nl$ and the point, where $u$ achieves maximum, is $l_{\max}(n)=\be n$.
The corresponding values of utility and tax are
\[
u_{\max}(n)=\al +\frac12(\be n)^2, \quad
t_{\max}(n)=-\al+(1-\be)\be n^2.
\]
Equalising total tax to zero yields
\[
\al =(1-\be) \be \int_0^{\infty} n^2 f(n) dn=(1-\be)\be  \E N^2,
\]
where we denoted by $N$ the random variable of skills with the density $f(n)$
and by $\E$ the corresponding expectation. Thus
\begin{equation}
\label{explin}
U=\int_0^{\infty}  u_{\max}(n) f(n) \, dn =\al + \frac12 \be^2 \E N^2
=\be (1-\frac12 \be) \E N^2.
\end{equation}

Confirming again the result of the previous section we see that maximum of $U$
equals $\E N^2/2$ and occurs at $\be=1$.

Next we calculate the variance. Since the variance $\si^2_u$ of the utility
is not changed by an additive constant, we can calculate it with
$\al=0$ yielding directly
\begin{equation}
\label{varlin}
\si_u^2=\frac14 \be^4 [\E N^4-(\E N^2)^2]=\frac14 \be^4 \si^2(N^2),
\end{equation}
where $\si^2(N^2)$ is the variance of the random variable $N^2$. Thus $\si_u$ is also increasing with $\be$
and takes its maximum value $\si(N^2)/2$ at $\be=1$.

As is the standard approach in a two-criteria optimisation, we can state the two dual problems,
one of minimising $\si_u$ under a given level of $U$, and another of maximising $U$
under a given level of $\si_u$. These two problems usually lead to one and the same
equation linking optimal values of criteria, the corresponding set of solutions being
referred to as the efficient frontier.

In the present simple case with only one parameter $\be$, the total utility and its variance
are seen to lie on the curve given in parametric form by
\eqref{varlin} and \eqref{explin}. Excluding $\be$ yields the equation connecting expected utility
$U$ and its standard deviation $\si_u$:
\begin{equation}
\label{expvar}
\left(\frac{2\si_u}{\si(N^2)}\right)^{1/2}-\frac{\si_u}{\si(N^2)}-\frac{U}{\E N^2}=0.
\end{equation}

In terms of normalised variables
\[
\tilde \si_u=\si_u/\si(N^2), \quad \tilde U=U/\E N^2
\]
it rewrites as
\begin{equation}
\label{expvar1}
 2\tilde \si_u=(\tilde \si_u +\tilde U)^2,
 \end{equation}
 which is seen to be a parabola. We are interested in its part with $\tilde \si_u\ge 0$,
 which is a concave curve lying in the first quadrant of the $(\tilde U, \tilde \si_u)$-plane and
 crossing $\{\tilde U=0\}$-axis at the points $\tilde \si_u=0$ and $\tilde \si_u=2$.
 The maximum of $\tilde U$ is $1/2$
 that occurs at $\tilde \si =1/2$. Since we are interested in minimal $\si$,
 the part of the parabola with $\tilde \si>1/2$
 is of no interest. Moreover, this part corresponds to $\be>1$, which are not allowed.
So only the part with $\tilde \si \in [0,1/2]$ represents the real {\it efficient frontier},
 where optimal solutions lie.

In order to choose a point on the efficient frontier of a two-criteria optimization, the usual approach
is to specify a new utility function of the form $V=U-c\si_u$, where $c$ is the key parameter for the decision making.
In the financial context this parameter measures the risk tolerance (or risk aversion). In the present social choice setting,
$1/c$ designates the growth of social inequality, which is acceptable for a unit growth of the total welfare.
This is of course an exogenously given parameter that must be fed to the mathematical model by public authorities.

For an example, assume that our distributions are normalised so that $\tilde U=U$, $\tilde\si_u=\si_u=\si$.
Thus, given $c$, we have to maximise $V=U-c\si$ on the efficient frontier given by \eqref{expvar1}.
Geometrically, this means to find the point on this frontier, which is tangent to a line parallel to $\{U=c\si\}$.
Analytically, this means to find $\be\in [0,1]$ yielding
\[
\max \{U-c\si: 2\si=(\si+U)^2, \,\, \si\in [0,1/2]\}.
\]
It is straightforward to see that maximum in this problem is achieved on
\[
\si=\frac{1}{2(c+1)^2}, \quad  U=\frac{2c+1}{2(c+1)^2},
\]
which, by \eqref{varlin} (recall that we assumed $\si(N^2)=1$),  corresponds to
$\be=\be(c)=(1+c)^{-1}$. Thus we see that the optimal $\be(c)$ can take any value
from $[0,1]$ depending on the choice of $c$.
For $c=0$ we return to the maximization of $U$ yielding $\be(c)=1$;
 and while $c$ increases to infinity, $\be(c)$ tends to zero.

\begin{remark}
Of course the same result can be obtained by searching $\max\{U-c\si\}$ directly over all $\be\in [0,1]$.
\end{remark}

\section{New paradigm: two bracket case}
\label{newpartwobrac}

Let us extend the analysis of the previous section to the piecewise linear tax
with two brackets, that is, with the tax function \eqref{eqtaxonekink}.

Recall that the budget constraint yields the value for $\al$ given by
\eqref{eqaltwo1} and \eqref{eqaltwo2}. To calculate the variance $\si_u^2$ of the utility,
we again take into account that it does not depend on $\al$.

In case $\be_1\le \be_2$ we have:

\[
\si_u^2=\int_0^{n_3} \frac14 \be_1^4 n^4 f(n) dn
+\int_{n_3}^{\infty} [\frac12 \be_2^2 n^2 -y_1(\be_2-\be_1)]^2 f(n) dn -(U-\al)^2,
\]
\[
U=\al +\int_0^{n_3} \frac12 \be_1^2 n^2 f(n) dn
+\int_{n_3}^{\infty} [\frac12 \be_2^2 n^2 -y_1(\be_2-\be_1)] f(n) dn
\]
\[
=\int_0^{n_3} \frac12 (\be_1-\frac12 \be_1^2) n^2 f(n) dn
+\int_{n_3}^{\infty} (\be_2-\frac12 \be_2^2) n^2 f(n) dn,
\]
where $\al$ is given by the formula (\ref{eqaltwo1}).\\
And in case $\be_1\ge \be_2$ we have:

\[
\si_u^2=\int_0^{n_1} \frac14 \be_1^4 n^4 f(n) dn+\int_{n_1}^{n_2} \left[\be_1 y_1-\frac{y_1^2}{2n^2}\right]^2 f(n) dn
\]
\[
+\int_{n_2}^{\infty} [\frac12 \be_2^2 n^2 -y_1(\be_2-\be_1)]^2 f(n) dn -(U-\al)^2,
\]
\[
U=\al+\int_0^{n_1} \frac12 \be_1^2 n^2 f(n) dn+\int_{n_1}^{n_2} \left[\be_1 y_1-\frac{y_1^2}{2n^2}\right] f(n) dn
+\int_{n_2}^{\infty} [\frac12 \be_2^2 n^2 -y_1(\be_2-\be_1)] f(n) dn
\]
\[
=\int_0^{n_1} (\be_1-\frac12 \be_1^2)n^2 f(n) dn
\]
\[
+\int_{n_1}^{n_2} (y_1-\frac{y_1^2}{2n^2}) f(n) dn
+\int_{n_2}^{\infty}(\be_2-\frac12 \be_2^2) n^2 f(n) dn,
\]
where $\al$ is given by the formula (\ref{eqaltwo2}).\\
As in the previous section the standard problems with two criteria is either maximise $U$ given $\si_u$
or minimise $\si_u$ given $U$. Both can be achieved by the analysis of the Lagrange function $V = U - c\sigma_{u}$. 

As we mentioned in introduction, analytic solution is not easy here and  shall perform numeric calculations
assuming that $N$  is uniformly distributed over $[0,10]$, that is, $f(n)=0.1$ for $n\in [0,10]$ and $f(n)=0$ otherwise. It is worth noting that governments rarely use fractional interest rates, so the precision suggested below will satisfy the necessary practical accuracy.

In this analysis, we have pairs of $\beta_{1}$ and $\beta_{2}$ values ranging from 0.01 to 1.00 with a step size of 0.01. In addition to the tax rate parameters, the model incorporates the breakpoint parameter $y_{1}$, which varies from 0.01 to 0.1 with a step size of 0.01. Consequently, the grid size is 10 x 100 x 100.

The issue remains unchanged from the previous paragraph. The objective is to maximise the value of $V = U - c\sigma_{u}$. Rather than finding the tangent to the graph, according to  Remark 1, the optimum will be identified by directly locating the maximum value across the entire grid.

\begin{table}[h!]
\centering
\begin{tabular}{||c c c c c c c ||}
 \hline
 $c$ & $\beta_{1}$ & $\beta_{2}$ & $y_{1}$ & $V$ & $U$ & $\sigma_{u}$ \\ [0.5ex]
 \hline\hline
 0.1 & 0.95 & 0.92 & 0.1 & 15.2982 & 16.5600 & 12.6174\\
 0.2 & 0.91 & 0.85 & 0.1 & 14.1376 & 16.2917 & 10.7705 \\
 0.3 & 0.87 & 0.79 & 0.1 & 13.1407 & 15.9318 & 9.3037 \\
 0.4 & 0.84 & 0.74 & 0.1 & 12.2749 & 15.5403 & 8.1634 \\
 0.5 & 0.81 & 0.69 & 0.1 & 11.5167 & 15.0655 & 7.0976\\ [1ex]
 \hline
\end{tabular}
\caption{Optimal solutions .}
\label{table:1}
\end{table}

Table 1 illustrates the optimal parameter values that yield the maximum value of $V$, with $c$ representing a known constant. As can be observed, the tax rate prior to the specified breakpoint is less than that subsequent to it. This indicates that the novel approach has effectively fulfilled the previously stated objective of fair taxation.

The partial values presented in the Table 1 illustrate that the relationship between $U$ and $\sigma_{u}$ follows a curve that resembles that of the linear case. Following the simulations for a larger number of values of $c$, a graph of $U(\sigma_{u})$ dependence was plotted(Fig. 1). The form of the obtained function closely resembles a portion of the parabola obtained earlier.

 \begin{figure}[htbp]
\centerline{\includegraphics[width=11cm, height=8cm]{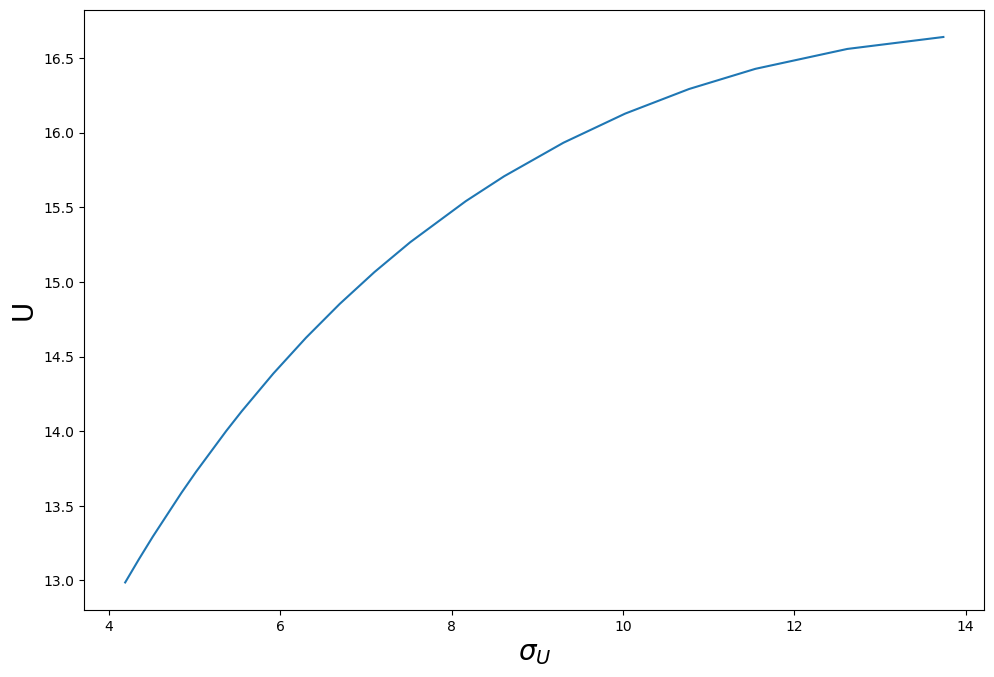}}
\caption{$U(\sigma_{u})$}
\label{fig}
\end{figure}

It is important to note that the solution to the optimisation problem, which minimises the variance $\sigma_{u}$ at a given level of mathematical expectation $U$, also results in a graph of the same form. This demonstrates the equivalence of these approaches.
\section{Logarithmic utility function}
\label{seclogut}
 In this section we will consider the model with the linear taxation (\ref{eqtaxlin}) and the following utility function:
 \begin{equation}
     u(c,l)=\ln(c)+A\, \ln(1-l),
 \end{equation}
where $c=c(y)=y-t(y),\, y=nl$.
\\By differentiating, we get:
\begin{equation*}
        l_{\max}(n)=\frac{A \be  n-\al }{(A+1) \be  n} \II ,
 \end{equation*}
where $\I_X$ is the indicator function of X.
\\Then the corresponding value of the utility function is
\begin{equation*}
    u_{\max}(n)=A\, \ln (\al )+
\left(k(A)+A\, \ln \left(1 +\frac{n}{\del }\right)+
\ln \left(\frac{\del +n}{n}\right)\right) \II ,
\end{equation*}
where $k(A)=A\, \ln (A)-(A+1)\, \ln (A+1)$.\\
Hence by definition of mathematical expectation and variation, we can get:
\begin{equation}
\begin{gathered}
U=A\, \ln (\al )+\E\left( \left(k(A)+A\, \ln \left(1 +\frac{N}{\del }\right)+
\ln \left(\frac{\del +N}{N}\right)\right)\III\right),\\
\si_u^2=
k^2(A) \si^2(\III)+
\si^2\left(
\left(A\, \ln \left(1 +\frac{N}{\del }\right)+
\ln \left(\frac{\del +N}{N}\right)\right) \III
\right)+\\
+2k(A)\E(\I_{[0,\frac{\al}{A \be}]}(N))
\E\left(
\left(A\, \ln \left(1 +\frac{N}{\del }\right)+
\ln \left(\frac{\del +N}{N}\right)\right) \III
\right)
\end{gathered}
\end{equation}
Consider the case when N is uniformly distributed over [0, s]. Then the expectation and the variance can be written explicitly. These  explicit forms are not important for further discussion, so we will not write them out. But we mention that they have the form
\begin{equation}
    \begin{gathered}
        U=A\ \,ln(\al)+F_1(A,\sig)\\
        \si_u=F_2(A,\sig),
    \end{gathered}
    \label{common form}
\end{equation}
where $F_1, \, F_2$ are some known functions.\\
 Also we assumed that the tax is fully distributed (\ref{eqtaxdist}). In the case under consideration, this condition takes the form
 \begin{equation}
\sig=\frac{A}{1-\beta}\left( 1+ A\beta  - \sqrt{(A+1) \beta(2+(A-1)\beta)}\right)
\label{sig(be)}
\end{equation}
Substituting equation (\ref{sig(be)}) into equation (\ref{common form}) we get

\begin{equation}
    \begin{gathered}
        U=A\, \ln\left(\al(A,\be,s)\right)+G_1(A,\be)\\
        \si_u=G_2(A,\be),
    \end{gathered}
    \label{common form 2}
\end{equation}
where $G_i=F_i\left(A,\sig(A,\be)\right), \, i=1,2$.

Let's plot the dependence of $U$ on $\be$ for fixed $A$ and $s$. They are qualitatively the same for different parameters $A$ and $s$. Consider as an example $s = 10^{12}, A = 1$. Then the graph is

{
    \centering
    \includegraphics[width=11cm, height=8cm]{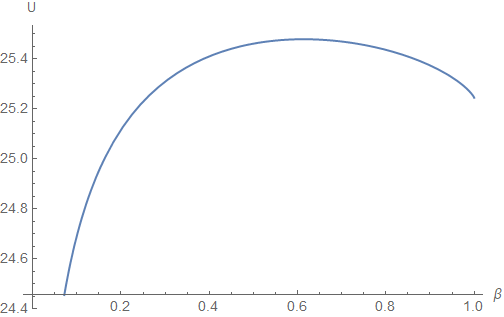}
    \captionof{figure}{$U(\be)$}
    \label{U(be)}
}

The maximum is clearly visible. Its existence and uniqueness can be proved analytically, but we omit this step to shorten the narrative.
The argument and the value of $U$ are numerically found. In the example under consideration
$$(\beta,U)_{max}=(0.6138,25.4788)$$
Let's build graphs of $\si_{u}(\be)$ with a fixed $A$. Qualitatively, the graph is the same for different $A$. For example, if $A=1$ we get

{
    \centering
    \includegraphics[width=11cm, height=8cm]{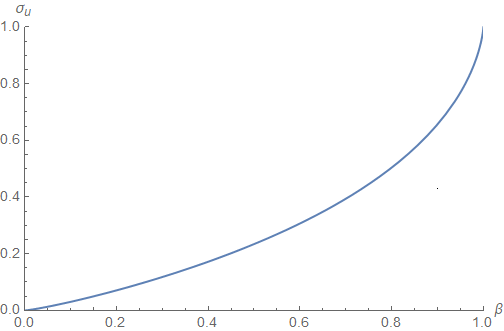}
    \captionof{figure}{$\si_{u}(\be)$}
    \label{si(be)}
}

It is clear that for fixed $A$ and $s$ we can get the parametric plot of $U(\si_u)$. Consider as an example $s = 10^{12}, A = 1$. Then the graph is

{
    \centering
    \includegraphics{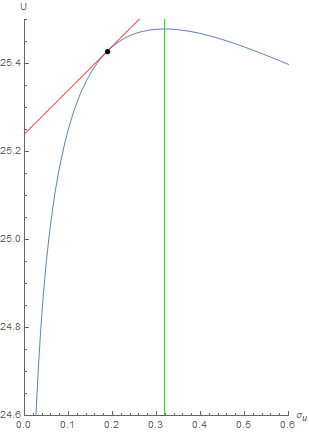}
    \captionof{figure}{$\si_{u}(\be)$}
    \label{U(si)}
}
If $U$ is fixed, it is more logical to choose $\be$, which has a lower standard deviation $\si_u$. That is why we do not consider the part of the graph that is to the right of the green vertical line, which intersects the blue parametric plot at the point with the maximum value of $U$.\\
Note that the problem of minimizing the function $V=U-c\si_u$ can
be implemented both numerically and
graphically. For example, if $c=1$ then the optimal values are
$$(\be,U,\sigma_{u})_{opt}=(0.427,25.428,0.188)$$

\end{document}